\documentclass{birkmult}
 \newtheorem{thm}{Theorem}[section]
 
 \newtheorem{lem}[thm]{Lemma}
 \newtheorem{prop}[thm]{Proposition}
 \theoremstyle{definition}
 
 \theoremstyle{remark}

 \numberwithin{equation}{section}

\begin{document}

 \newtheorem{rema}[thm]{Remark}{\hspace*{4mm}}
 \newtheorem{definition}[thm]{Definition}
%
%
%
%
%
%
%
%
%
\title {The Square of Opposition in Orthomodular Logic}
\author{H. Freytes}

\address{%
Universita degli Studi di Cagliari\\
Via Is Mirrionis 1\\
09123, Cagliari \\
Italia\\
Instituto Argentino de Matem\'atica\\
Saavedra 15\\
Buenos Aires\\
Argentina}

\email{hfreytes@gmail.com}

\author{C. de Ronde}
\address{Center Leo Apostel \br
Krijgskundestraat 33\br 1160 Brussels\br Belgium}
\email{cderonde@vub.ac.be}
\author{G. Domenech}
\address{Instituto de Astronom\'{\i}a y F\'{\i}sica del Espacio \br
CC 67, Suc 28\br 1428 Buenos Aires\br Argentina}
\email{domenech@iafe.uba.ar}


\date{January 1, 2004}

\begin{abstract}
In Aristotelian logic, categorical propositions are divided in
Universal Affirmative, Universal Negative, Particular Affirmative
and Particular Negative. Possible relations between two of the
mentioned type of propositions are encoded in the square of
opposition. The square expresses the essential properties of monadic
first order quantification which, in an algebraic approach, may be
represented taking into account monadic Boolean algebras. More
precisely, quantifiers are considered as modal operators acting on a
Boolean algebra and the square of opposition is represented by
relations between certain terms of the language in which the
algebraic structure is formulated. This representation is sometimes
called the modal square of opposition. Several generalizations of
the monadic first order logic can be obtained by changing the
underlying Boolean structure by another one giving rise to new
possible interpretations of the square.
\end{abstract}

\maketitle


\section*{Introduction}
In Aristotelian logic, categorical propositions are divided into
four basic types: {\it Universal Affirmative}, {\it Universal
Negative}, {\it Particular Affirmative} and {\it Particular
Negative}. The possible relations between each two of the mentioned
propositions are encoded in the famous {\it Square of Opposition}.
The square expresses the essential properties of the monadic first
order quantifiers $\forall$, $\exists$. In an algebraic approach,
these properties can be represented within the frame of monadic
Boolean algebras \cite{HAL}. More precisely, quantifiers are
considered as modal operators acting on a Boolean algebra while the
Square of Opposition is represented by relations between certain
terms of the language in which  the algebraic structure is
formulated. This representation is sometimes called {\it Modal
Square of Opposition} and is pictured as follows:

\vspace{0.6cm}

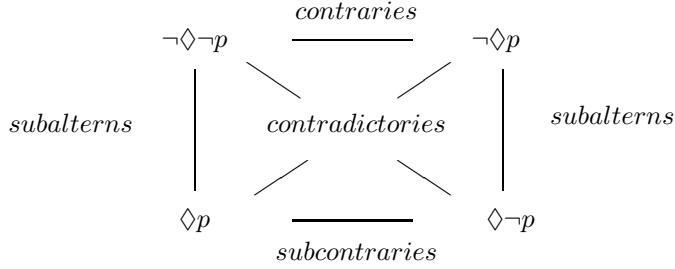
\begin{center}
\unitlength=1mm
\begin{picture}(20,20)(0,0)
\put(3,16){\line(3,0){16}} \put(-10,12){\line(0,-2){16}}
\put(3,-8){\line(1,0){16}} \put(31,12){\line(0,-2){16}}

\put(-10,16){\makebox(0,0){$\neg \Diamond \neg p$}}
\put(30,16){\makebox(0,0){$\neg \Diamond p$}}
\put(-10,-8){\makebox(0,0){$\Diamond p$}}
\put(32,-8){\makebox(0,0){$\Diamond \neg p$}}

\put(4,20){\makebox(15,0){$contraries$}}
\put(-24,5){\makebox(-5,0){$subalterns$}}
\put(14,-13){\makebox(-5,2){$subcontraries$}}
\put(46,5){\makebox(-1,2){$subalterns$}}
\put(12,4){\makebox(-1,2){$contradictories$}}

\put(-3,13){\line(3,-2){7}} \put(17,0){\line(3,-2){7}}

\put(24,13){\line(-3,-2){7}} \put(5,0){\line(-3,-2){7}}

\end{picture}
\end{center}

\vspace{1.5cm}

The interpretations given to $\Diamond$ from different modal logics
determine the corresponding versions of the modal Square of
Opposition. By changing the underlying Boolean structure we obtain
several generalizations of the monadic first order logic (see for
example \cite{HAJ}). In turn, these generalizations give rise to new
interpretations of the Square.

The aim of this paper is to study the Square of Opposition in an
orthomodular structure enriched with a monadic quantifier known as
{\it Boolean saturated orthomodular lattice } \cite{DFD1}.  The
paper is structured as follows. Section 1 contains generalities on
orthomodular lattices. In Section 2, the physical motivation for the
modal enrichment of the orthomodular structure is presented. In
Section 3 we formalize the concept of classical consequence with
respect to a property of a quantum system. Finally, in Section 4,
logical relationships between the propositions embodied in a square
diagram are studied in terms of  classical consequences and
contextual valuations.

\section{Basic Notions}

We recall from \cite{Bur},  \cite{KAL} and \cite{MM} some notions of
universal algebra and lattice theory that will play an important
role in what follows.  Let ${\mathcal{L}}=\langle
{\mathcal{L}},\lor,\land, 0, 1\rangle$ be a bounded lattice. An
element $c\in \mathcal{L}$ is said to be a {\it complement} of $a$
iff $a\land c = 0$ and $a\lor c = 1$.  Given $a, b, c$ in
$\mathcal{L}$, we write: $(a,b,c)D$\ \ iff $(a\lor b)\land c =
(a\land c)\lor (b\land c)$; $(a,b,c)D^{*}$ iff $(a\land b)\lor c =
(a\lor c)\land (b\lor c)$ and $(a,b,c)T$\ \ iff $(a,b,c)D$,
(a,b,c)$D^{*}$ hold for all permutations of $a, b, c$. An element
$z$ of a lattice $\mathcal{L}$ is called {\it central} iff for all
elements $a,b\in \mathcal{L}$ we have $(a,b,z)T$ and $z$ is
complemented. We denote by $Z(\mathcal{L})$ the set of all central
elements of $\mathcal{L}$ and it is called the {\it center} of
$\mathcal{L}$.

A {\it  lattice with involution} \cite{Ka} is an algebra $\langle
\mathcal{L}, \lor, \land, \neg \rangle$ such that $\langle
\mathcal{L}, \lor, \land \rangle$ is a  lattice and $\neg$ is a
unary operation on $\mathcal{L}$ that fulfills the following
conditions: $\neg \neg x = x$ and $\neg (x \lor y) = \neg x \land
\neg y$. An {\it orthomodular lattice} is an algebra $\langle
{\mathcal{L}}, \land, \lor, \neg, 0,1 \rangle$ of type $\langle
2,2,1,0,0 \rangle$ that satisfies the following conditions

\begin{enumerate}
\item
$\langle {\mathcal{L}}, \land, \lor, \neg, 0,1 \rangle$ is a bounded
lattice with involution,

\item
$x\land  \neg x = 0 $.

\item
$x\lor ( \neg x \land (x\lor y)) = x\lor y $

\end{enumerate}

We denote by ${\mathcal OML}$ the variety of orthomodular lattices.
Let $L$ be an orthomodular lattice and $a,b \in L$. Then $a$
commutes with $b$ if and only if $a = (a\land b) \lor (a \land \neg
b)$. A non-empty subset $A$ is called a {\it Greechie set} iff for
any three different elements of $A$, at least one of them commutes
with the other two. If $A$ is a Greechie set in $L$ then $\langle A
\rangle_L$, i.e. the sublattice generated by $A$, is distributive
\cite{GREE}. {\it Boolean algebras} are orthomodular lattices
satisfying the {\it distributive law} $x\land (y \lor z) = (x \land
y) \lor (x \land z)$. We denote by ${\bf 2}$ the Boolean algebra of
two elements. Let $A$ be a Boolean algebra. Then, as a consequence
of the application of the {\it maximal filter theorem} for Boolean
algebras, there always exists a Boolean homomorphism $f:A\rightarrow
{\bf 2}$.  If $\mathcal{L}$ is a bounded lattice then
$Z(\mathcal{L})$ is a Boolean sublattice of $\mathcal{L}$ {\rm
\cite[Theorem 4.15]{MM}}.

\section{Modal Propositions about Quantum Systems}

In the usual terms of quantum logic \cite{ByvN, Jauch68}, a property
of a system is related to a subspace of the Hilbert space ${\mathcal
H}$ of its (pure) states or, analogously, to the projector operator
onto that subspace. A physical magnitude ${\mathcal M}$ is
represented by an operator $\bf M$ acting over the state space. For
bounded self-adjoint operators, conditions for the existence of the
spectral decomposition ${\bf M}=\sum_{i} a_i {\bf P}_i=\sum_{i} a_i
|a_i\rangle\langle a_i|$ are satisfied. The real numbers $a_i$ are
related to the outcomes of measurements of the magnitude ${\mathcal
M}$ and projectors $|a_i\rangle\langle a_i|$ to the mentioned
properties. Thus, the physical properties of the system are
organized in the lattice of closed subspaces ${\mathcal L}({\mathcal
H})$. Moreover, each self-adjoint operator $\bf M$  has associated a
Boolean sublattice $W_{\bf{M}}$ of $L({\mathcal H})$ which we will
refer to as the spectral algebra of the operator $\bf M$. More
precisely, the family $\{{\bf  P_i}\}$ of projector operators is
identified as elements of $W_{\bf{M}}$. Assigning values to a
physical quantity ${\mathcal M}$ is equivalent to establishing a
Boolean homomorphism $v: W_{\bf{M}} \rightarrow {\bf 2}$ which we
call {\it contextual valuation}. Thus, we can say that it makes
sense to use the ``classical discourse''
---this is, the classical logical laws are valid--- within the
context given by ${\bf M}$.

Modal interpretations of quantum mechanics \cite{Dieks1, Dieks2,
VF2} face the problem of finding an objective reading of the
accepted mathematical formalism of the theory, a reading ``in terms
of properties possessed by physical systems, independently of
consciousness and measurements (in the sense of human
interventions)''\cite{Dieks2}. These interpretations intend to
consistently include the possible properties of the system in the
discourse establishing a new link between the state of the system
and the probabilistic character of its properties, namely,
sustaining that the interpretation of the quantum state must contain
a modal aspect. The name modal interpretation was used for the first
time by B. van Fraassen \cite{VF1} following modal logic, precisely
the logic that deals with \emph{possibility} and \emph{necessity}.
The fundamental idea is to interpret ``the formalism as providing
information about properties of physical systems''. A physical
property of a system is ``a definite value of a physical quantity
belonging to this system; i.e., a feature of physical reality''
\cite{Dieks1} and not a mere measurement outcome. As usual, definite
values of physical magnitudes correspond to yes/no propositions
represented by orthogonal projection operators acting on vectors
belonging to the Hilbert space of the (pure) states of the system
\cite{Jauch68}.

The proposed modal system  \cite{DFD1, DFD2} over which  we now
develop an interpretation of the Square of Opposition, is based on
the study of the  ``{\it classical consequences}'' that result from
assigning values to a physical quantity. In precise terms,  we
enriched the orthomodular structure with a modal operator taking
into account the following considerations:\\

1) Propositions about the properties of the physical system are
interpreted in the orthomodular lattice of closed subspaces of
$\mathcal{H}$. Thus, we retain this structure in our extension.\\

2) Given a proposition about the system, it is possible to define a
context from which one can predicate with certainty about it
together with a set of propositions that are compatible with it and,
at the same time, predicate probabilities about the other ones (Born
rule). In other words, one may predicate truth or falsity of all
possibilities at the same time, i.e. possibilities allow an
interpretation in a Boolean algebra. In rigorous terms, for each
proposition $p$, if we refer with $\Diamond p$ to the possibility of
$p$, then $\Diamond p$ will be a central element of a orthomodular
structure.\\

3) If $p$ is a proposition about the system and $p$ occurs, then it
is trivially possible that $p$ occurs. This is expressed as $p \leq
\Diamond p$. \\

4) Let $p$ be a property appertaining to a context ${\mathcal M}$.
Assuming that $p$ is an actual property (for example the result of a
filtering measurement) we may derive from it  a set of propositions
(perhaps not all of them encoded in the original Hilbert lattice of
the system) which we call {\it classical consequences}. For example,
let $q$ be another property of the system and assign to $q$ the
probability $prob(q) = r$ via the Born rule. Then equality $prob(q)
= r$ will be considered as a  classical consequence of $p$. In fact,
the main characteristic of this type of classical consequences is
that it is possible to simultaneously predicate the truth of all of
them (and the falsity of their negations) whenever $p$ is true. The
formal representation of the concept of classical consequence is the
following: A proposition $t$ is a classical consequence of $p$ iff
$t$ is in the center of an orthomodular lattice containing $p$ and
satisfies the property $p\leq t$. These classical consequences are
the same ones as those which would be obtained by considering the
original actual property $p$ as a possible one $\Diamond p$.
Consequently $\Diamond p$ must precede all classical consequences of
$p$. This is interpreted in the following way: $\Diamond p$ is the
smallest
central element greater than $p$.\\

From consideration 1, it follows that the original orthomodular
structure is maintained. The other considerations are satisfied if
we consider a modal operator $\Diamond$ over an orthomodular lattice
$\mathcal{L}$  defined as $$\Diamond a = Min \{z\in
Z({\mathcal{L}}): a\leq z \}$$ with $Z({\mathcal{L}})$ the center of
$\mathcal{L}$. When this minimum exists for each $a\in \mathcal{L}$
we say that $\mathcal{L}$ is a {\it Boolean saturated orthomodular
lattice}. We have shown that this enriched orthomodular structure
can be axiomatized by equations conforming a variety denoted by
${\mathcal OML}^\Diamond$ \cite{DFD1}. More precisely, each element
of ${\mathcal OML}^\Diamond$ is an algebra $ \langle {\mathcal{L}},
\land, \lor, \neg, \Diamond, 0, 1 \rangle$ of type $ \langle 2, 2,
1, 1, 0, 0 \rangle$ such that $ \langle {\mathcal{L}}, \land, \lor,
\neg, 0, 1 \rangle$ is an orthomodular lattice and $\Box$ satisfies
the following equations:

\begin{enumerate}

\item[S1]
$x \leq \Diamond x$ \hspace{3cm} S5 \hspace{0.2cm} $y = (y\land
\Diamond x) \lor (y \land \neg \Diamond x)$

\item[S2]
$\Diamond 0 = 0$  \hspace{3.1 cm} S6 \hspace{0.2cm} $\Diamond (x
\land \Diamond y ) = \Diamond x \land \Diamond y $

\item[S3]
$\Diamond \Diamond x = \Diamond x$  \hspace{2.5cm} S7 \hspace{0.2cm}
$\neg \Diamond x \land \Diamond y   \leq \Diamond (\neg x \land (y
\lor x)) $

\item[S4]
$\Diamond (x \lor y) = \Diamond x  \lor  \Diamond y$

\end{enumerate}

Orthomodular complete lattices are examples of Boolean saturated
orthomodular lattices. We can embed each orthomodular lattice
$\mathcal{L}$ in an element $\mathcal{L}^{\Diamond} \in  {\mathcal
OML}^\Diamond$ {see \rm \cite[Theorem 10]{DFD1}}. In general,
$\mathcal{L}^{\Diamond}$ is referred as a {\it modal extension of
$\mathcal{L}$}. In this case we may see the lattice $\mathcal{L}$ as
a subset of $\mathcal{L}^{\Diamond}$.

\section{Modal Extensions and Classical Consequences}
We begin our study of the Square of Opposition analyzing the
classical consequences that can be derived from a proposition about
the system. This idea was suggested in the condition 4 of the
motivation of the structure ${\mathcal OML}^\Diamond$. In what
follows we express the notion of classical consequence as a formal
concept in ${\mathcal OML}^\Diamond$. We first need the following
technical results:

\begin{prop}\label{GREE}
Let $L$ be an orthomodular lattice. If $A$ is a Greechie set in $L$
such that for each $a\in A, \neg a \in A$ then, $\langle A \rangle_L
$ is Boolean sublattice.
\end{prop}

\begin{proof}
It is well known from {\rm \cite{GREE}} that $\langle A \rangle_L $
is a distributive sublattice of $L$. Since distributive orthomodular
lattices are Boolean algebras, we only need to see that $\langle A
\rangle_L $ is closed by $\neg$. To do that we use induction on the
complexity of terms of the subuniverse generated by  $A$. For
$comp(a) = 0$, it follows from the fact that $A$ is closed by
negation. Assume validity for terms of the complexity less than $n$.
Let $\tau$ be a term such that $comp(\tau)= n$. If $\tau = \neg
\tau_1$ then $\neg \tau \in \langle A \rangle_L$ since $\neg \tau =
\neg \neg \tau_1 = \tau_1$ and $\tau_1 \in \langle A \rangle_L$. If
$\tau = \tau_1 \land \tau_2$, $\neg \tau = \neg \tau_1 \lor \neg
\tau_2$. Since $comp(\tau_i) < n$, $\neg \tau_i \in \langle A
\rangle_L$ for $i= 1,2$ resulting $\neg \tau \in \langle A
\rangle_L$. We use the same argument in the case $\tau = \tau_1 \lor
\tau_2$. Finally $\langle A \rangle_L$ is a Boolean sublattice.
\end{proof}

Since the center $Z({\mathcal L}^\Diamond)$ is a Boolean algebra, it
represents a fragment of discourse added to ${\mathcal L}$ in which
the laws of the classical logic are valid. Thus the modal extension
${\mathcal L} \hookrightarrow {\mathcal L}^\Diamond $ is a structure
that rules the mentioned fragment of classical discourse and the
properties about a quantum system encoded in ${\mathcal L}$. Let $W$
be a Boolean sub-algebra of ${\mathcal L}$ (i.e. a context). Note
that $W \cup Z({\mathcal L}^\Diamond)$ is a Greechie set closed by
$\neg$. Then by Proposition \ref{GREE}, $\langle W \cup Z({\mathcal
L}^\Diamond) \rangle_{{\mathcal L}^\Diamond}$ is a Boolean
sub-algebra of $Z({\mathcal L}^\Diamond)$. This represents the
possibility to fix a context  and compatibly add a fragment of
classical discourse. We will refer to the Boolean algebra
$W^{\Diamond} = \langle W \cup Z({\mathcal L}^\Diamond)
\rangle_{{\mathcal L}^\Diamond}$ as a {\it classically expanded
context}. Taking into account that assigning values to a physical
quantity $p$ is equivalent to fix a context $W$ in which $p\in W$
and establish a Boolean homomorphism $v: W \rightarrow {\bf 2}$ such
that $v(p)=1$, we give the following definition of classical
consequence.

\begin{definition}\label{CLASCONS0}
{\rm Let ${\mathcal L}$ be an orthomodular lattice, $p\in {\mathcal
L}$ and ${\mathcal L}^\Diamond \in {\mathcal OML}^\Diamond$ a modal
extension of $\mathcal{L}$. Then $z \in Z({\mathcal L}^\Diamond) $
is said to be a {\it classical consequence} of $p$ iff for each
Boolean sublattice $W$ in $\mathcal{L}$ (with $p\in W$) and each
Boolean valuation $v:W^{\Diamond} \rightarrow {\bf 2}$, $v(z) = 1$
whenever $v(p)=1$. }
\end{definition}

$v(p) = 1$ implies $v(z) = 1$ in Definition \ref{CLASCONS0} is a
relation in a classically expanded context that represents, in an
algebraic way, the usual logical consequence of $z$ from $p$. We
denote by $Cons_{\mathcal{L}^\Diamond}(p)$ the set of classical
consequences of $p$ in the modal extension ${\mathcal L}^\Diamond$.

\begin{prop}\label{CLASCONS}
Let ${\mathcal L}$ be an orthomodular lattice, $p\in {\mathcal L}$
and ${\mathcal L}^\Diamond \in {\mathcal OML}^\Diamond$ a modal
extension of ${\mathcal L}$. Then we have that $$Cons_{{\mathcal
L}^\Diamond}(p) = \{z\in Z({\mathcal L}^\Diamond): p \leq z \} =
\{x\in Z({\mathcal L}^\Diamond): \Diamond p \leq z \} $$
\end{prop}

\begin{proof}
$\{z\in Z({\mathcal L}^\Diamond): p \leq z \} = \{z\in Z({\mathcal
L}^\Diamond): \Diamond p \leq z \}$ follows from definition of
$\Diamond$. The inclusion $\{z\in Z({\mathcal L}^\Diamond): \Diamond
p \leq z \} \subseteq Cons_{{\mathcal L}^\Diamond}(p)$ is trivial.
Let $z\in Cons_{{\mathcal L}^\Diamond}(p)$ and suppose that $p \not
\leq z$. Consider the Boolean sub-algebra of ${\mathcal L}$ given by
$W = \{p, \neg p, 0,1 \}$. By the maximal filter theorem, there
exists a maximal filter $F$ in $W^\Diamond$ such that $p\in F$ and
$z \not \in F$. If we consider the quotient Boolean algebra
$W^\Diamond/F = {\bf 2}$, the natural Boolean homomorphism $f:
W^\Diamond \rightarrow {\bf 2}$ satisfies that  $f(p) = 1$ and $f(z)
= 0$, which is a contradiction. Hence $p \not \leq z$ and $z\in
Cons_{{\mathcal L}^\Diamond}(p)$.
\end{proof}

The equality $Cons_{{\mathcal L}^\Diamond}(p) = \{x\in Z({\mathcal
L}^\Diamond): \Diamond p \leq z \} $ given in Proposition
\ref{CLASCONS} states that the notion of classical consequence of
$p$ results independent of the choice of the context  $W$ in which
$p\in W$. In fact, each possible classical consequence $z \in
Z({\mathcal L}^\Diamond)$ of $p$ is only determined  by the relation
$\Diamond p \leq z $. Thus $Z({\mathcal L}^\Diamond)$ is a fragment
of the classical discourse added to ${\mathcal L}$ which allows to
``predicate'' classical consequences about the properties of the
system encoded in ${\mathcal L}$ independently of the context. It is
important to remark that the contextual character of the quantum
discourse is only avoided when we refer to ``classical consequences
of properties about the system'' and not when referring to the
properties in themselves, i.e. independently of the choice of the
context. In fact, in our modal extension, the discourse about
properties is genuinely enlarged, but the contextual character
remains a main feature of quantum systems even when modalities are
taken into account.

\section{Square of Opposition: Classical Consequences and Contextual Valuations}

In this section we analyze the relations between propositions
encoded in the scheme of the Square of Opposition in terms of the
classical consequences of a chosen  property about the quantum
system. To do this, we use the modal extension. Let $\mathcal{L}$ be
an orthomodular lattice, $p\in {\mathcal L}$ and ${\mathcal
L}^\Diamond \in {\mathcal OML}^\Diamond$ a modal extension of
$\mathcal{L}$. We first study the proposition $\neg \Diamond \neg p$
denoted by $\Box p$. Note that $\Box p = \neg \Diamond \neg p = \neg
Min\{z \in Z({\mathcal L}^{\Diamond}): \neg p \leq z  \} = Max
\{\neg z \in Z({\mathcal L}^{\Diamond}): \neg p \leq z \} = Max
\{\neg z \in Z({\mathcal L}^{\Diamond}): \neg z \leq p \}$.
Considering $t = \neg z$ we have that $$\Box p = \neg \Diamond \neg
p = Max\{t \in Z({\mathcal L}^{\Diamond}): t \leq p \}$$ When $W$ is
a Boolean sublattice of ${\mathcal L}$ such that $p\in W$ (i.e. we
are fixing a context containing $p$),  $\Box p$ is the greatest
classical proposition that implies $p$ in the  classically expanded
context $W^\Diamond$. More precisely, if $p$ is a consequence of
$z\in {\mathcal L}^{\Diamond}$ then $\Box p$ is consequence of $z$.

\begin{rema}
{\rm It is important to notice that Proposition \ref{CLASCONS}
allows us to refer to  classical consequences of a property of the
system independently of the chosen context. But in order to refer to
a property which is implied by a classical property, we need to fix
a context and consider the classically expanded context.}
\end{rema}

\begin{lem}\label{CLAS3}
Let ${\mathcal L}$ be an orthomodular lattice, $p \in {\mathcal L}$
and ${\mathcal L}^\Diamond$ be a modal extension of ${\mathcal L}$.
If $\Diamond p \land \Diamond \neg p = 0$ then $p\in Z({\mathcal
L})$.

\end{lem}

\begin{proof}
$\Diamond p$ and $\Diamond \neg p$ are central elements in
${\mathcal L}^\Diamond$. Taking into account that $1 = p\lor \neg p
\leq \Diamond p \lor \Diamond \neg p$, if $\Diamond p \land \Diamond
\neg p = 0$ then $\neg \Diamond p = \Diamond \neg p$ since the
complement is unique in $Z({\mathcal L}^\Diamond)$. Hence  $\Box p =
\neg \Diamond \neg p = \Diamond p$, $p\in Z({\mathcal L}^\Diamond)$
and $p\in Z({\mathcal L})$. 
\end{proof}

Now we can interpret the relation between propositions in the Square
of Opposition. In what follows we assume that ${\mathcal L}$ is an
orthomodular lattice and $p \in {\mathcal L}$ such that $p \not \in
Z({\mathcal L})$, i.e. $p$ is not a classical proposition in a
quantum system represented by ${\mathcal L}$. Let ${\mathcal
L}^\Diamond$ be a modal extension of ${\mathcal L}$, $W$ be a
Boolean subalgebra of ${\mathcal L}$, i.e. a context, such that
$p\in W$ and consider a classically expanded context $W^\Diamond$.

\begin{itemize}
\item
$\neg \Diamond \neg p \hspace{0.2cm} \underline{contraries}
\hspace{0.2cm} \neg \Diamond p $
\end{itemize}

\noindent $\neg \Diamond p = \neg \Diamond \neg \neg p = \Box \neg
p$. Thus, the \emph{contrary proposition} is the greatest classical
proposition that implies $p$, i.e. $\Box p$, with  respect to  the
greatest classical proposition that implies $\neg p$, i.e. $\Box
\neg p$,  in each possible classically expanded context containing
$p,\neg p$.

In the usual explanation, two propositions are contrary iff they
cannot both be true but can both be false. In our framework we can
obtain a similar concept of contrary propositions. Note that $\Box p
\land \Box \neg p \leq p \land \neg p = 0$. Thus, there is not a
maximal Boolean filter containing $\Box p$ and $\Box \neg p$. Hence
there is not a Boolean valuation $v:W^\Diamond \rightarrow {\bf 2}$
such that $v(\Box p) = v(\Box \neg p) = 1$, i.e. $\Box p$ and $\Box
\neg p$ ``cannot both be true'' in each possible classically
expanded context.

Since $p \not \in Z({\mathcal L})$, by Lemma \ref{CLAS3}, $\Diamond
p \land \Diamond \neg p \not = 0$. Then there exists a maximal
Boolean filter $F$ in $W^\Diamond$ containing $\Diamond p$ and
$\Diamond \neg p$. $\Box p \not \in F$ otherwise $\Box p \land
\Diamond \neg p \in F$ and $\Box p \land \Diamond \neg p = \Diamond
(\Box p \land \neg p) \leq \Diamond (p \land \neg p) = 0$ which is a
contradiction. With the same argument we can prove that $\Box  \neg
p \not \in F$. If we consider the natural homomorphism $v:
W^\Diamond \rightarrow W^\Diamond / F \approx {\bf 2}$ then $v(\Box
p) = v(\Box \neg p) = 0$, i.e. $\Box p$ and $\Box \neg p$ can both
be false.

\begin{itemize}
\item
$\Diamond p \hspace{0.2cm} \underline{subcontraries} \hspace{0.2cm}
\Diamond \neg p $
\end{itemize}
\noindent The \emph{sub-contrary proposition} is the smallest
classical consequence of $p$ with respect to the smallest classical
consequence of $\neg p$. Note that sub-contrary propositions do not
depend on the context.

In the usual explanation, two propositions are sub-contrary iff they
cannot both be false but can both be true. Suppose that there exists
a Boolean homomorphism $v: W^\Diamond \rightarrow {\bf 2}$ such that
$v(\Diamond p) = v(\Diamond \neg p) = 0$. Consider the maximal
Boolean filter given by $Ker(v)$. Since $Ker(v)$ is a maximal filter
in $W^\Diamond$, $p\in F_v$ or $\neg p\in F_v$. If $p\in F_v$ then
$v(\Diamond p) = 1$ which is a contradiction, if $\neg p\in F_v$
then $v(\Diamond \neg p) = 1$ which is a contradiction too. Hence
$v(\Diamond p) \not = 0$ or $v(\Diamond \neg p) \not = 0$, i.e. they
cannot both be false. Since $p \not \in Z({\mathcal L})$, by Lemma
\ref{CLAS3}, $\Diamond p \land \Diamond \neg p \not = 0$. Then there
exists a maximal Boolean filter $F$ in $W^\Diamond$ containing
$\Diamond p$ and $\Diamond \neg p$. Hence the Boolean homomorphism
$v: W^\Diamond \rightarrow W^\Diamond / F \approx {\bf 2}$ satisfies
that $v(\Diamond p) = v(\Diamond \neg p) = 1$, i.e. $\Diamond p$ and
$\Diamond \neg p$ can both be true.

\begin{itemize}
\item
$\neg \Diamond \neg p \hspace{0.2cm} \underline{subalterns}
\hspace{0.2cm} \Diamond p $ and $\neg \Diamond p \hspace{0.2cm}
\underline{subalterns} \hspace{0.2cm} \Diamond \neg p $
\end{itemize}
The notion of sub-contrary propositions is reduced to the relation
between  $\Box p$ and $\Diamond p$. The \emph{subaltern proposition}
is the greatest classical proposition that implies $p$ with respect
to the smallest classical consequence of $p$.

In the usual explanation, a proposition is subaltern of another one
called {\it superaltern}, iff it must be true when its superaltern
is true, and the superaltern must be false when the subaltern is
false. In our case  $\neg \Diamond \neg p = \Box p$ is superaltern
of $\Diamond p$ and $\neg \Diamond p = \Box \neg p$ is superaltern
of $\Diamond \neg p$. Since $\Box p \leq p \leq \Diamond p$, for
each valuation $v:W^\Diamond \rightarrow {\bf 2}$, if $v(\Box p) =
1$ then $v(\Diamond p) = 1$ and if  $v(\Diamond p) = 0$ then $v(\Box
p) = 0$.

\begin{itemize}
\item
$\neg \Diamond \neg p \hspace{0.2cm} \underline{contradictories}
\hspace{0.2cm} \Diamond \neg p $ and $\Diamond p \hspace{0.2cm}
\underline{contradictories} \hspace{0.2cm} \neg \Diamond p $
\end{itemize}
The notion of \emph{contradictory proposition} can be reduced to the
relation between  $\Diamond p$ and $\Box \neg p$. The contradictory
proposition  is the greatest classical proposition that implies
$\neg p$ with respect to the the smallest classical consequence of
$p$. In the usual explanation, two propositions are contradictory
iff they cannot both be true and they cannot both be false.  Due to
the fact that $ker(v)$ is a maximal filter in $W^\Diamond$,  each
maximal filter $F$  in $W^\Diamond$ contains exactly one of
$\{\Diamond p, \neg \Diamond p \}$ for each Boolean homomorphism
$v:W^\Diamond \rightarrow {\bf 2}$, $v(\Diamond p) = 1$ and $v(\neg
\Diamond p) = 0$ or $v(\Diamond p) = 0$ and $v(\neg \Diamond p) =
1$.  Hence, $\Diamond p$ and $\Diamond p$ cannot both be true and
they cannot both be false.

\end{document}